\documentclass[review]{elsarticle}

\usepackage{lineno,hyperref,mathtools}
\usepackage[letterpaper,top=2cm,bottom=2cm,left=3cm,right=3cm,marginparwidth=1.75cm]{geometry}

\journal{Journal of \LaTeX\ Templates}

\usepackage{amssymb}
\usepackage{bm}
\usepackage{color}
\usepackage{ulem}
\usepackage{amssymb}
\usepackage{amsthm}
\usepackage{multirow,booktabs}
\usepackage{cases}
\usepackage{threeparttable}

\newcommand{\F}{\mathbb{F}}
\newcommand{\Z}{\mathbb{Z}}
\newcommand{\C}{\mathcal{C}}
\newcommand{\lcm}{\mathrm{lcm}}
\newcommand{\ord}{\mathrm{ord}}

\newtheorem{theorem}{Theorem}

\newtheorem{lemma}[theorem]{Lemma}

\newtheorem{example}[theorem]{Example}
\begin{document}

	\begin{frontmatter}
		
		\title {Several classes of BCH codes of length $n=\frac{q^{m}-1}{2}$}
		\tnotetext[mytitlenote]{This research work is supported by the National Natural Science Foundation of China under Grant Nos. U21A20428 and 12171134.}

		\author[mymainaddress]{Mengchen Lian}
		\ead{lianmengchen0209@163.com}

		\author[mymainaddress]{Shixin Zhu\corref{mycorrespondingauthor}}
		\cortext[mycorrespondingauthor]{Corresponding author}
		\ead{zhushixinmath@hfut.edu.cn}
		
		\address[mymainaddress]{School of Mathematics, HeFei University of Technology, Hefei 230601, China}
		
		\begin{abstract}
			BCH codes are an important class of linear codes and find extensive utilization in communication and disk storage systems.
			This paper mainly analyzes the negacyclic BCH code and cyclic BCH code of length $\frac{q^m-1}{2}$. 
			For negacyclic BCH code, we give the dimensions of $\C_{(n,-1,\left\lceil \frac{\delta+1}{2}\right\rceil,0)}$ for $\delta =a\frac{q^m-1}{q-1},aq^{m-1}-1$($1\leq a <\frac{q-1}{2}$) and 
			$\delta =a\frac{q^m-1}{q-1}+b\frac{q^m-1}{q^2-1},aq^{m-1}+(a+b)q^{m-2}-1$ $(2\mid m,1\leq a+b \leq q-1$,$\left\lceil \frac{q-a-2}{2}\right\rceil\geq 1)$. 
			Furthermore, the dimensions of negacyclic BCH codes $\C_{(n,-1,\delta,0)}$ with few nonzeros and  $\C_{(n,-1,\delta,b)}$ with $b\neq 0$ are settled.
			For cyclic BCH code, we give the weight distribution of extended code $\overline{\C}_{(n,1,\delta,1)}$
			and the parameters of dual code $\C^{\perp}_{(n,1,\delta,1)}$, where $\delta_2\leq \delta \leq \delta_1$.
			
		\end{abstract}
		
	\end{frontmatter}
	
	\section{Introduction}
	Let $\C$ be an $[n,k,d]$ linear code over finite field $\F_q$ where $q$ is a prime power, 
	having dimension $k$ and  minimum distance $d$,
	its dual code is 
	\begin{center}
		$\C^{\perp}:=\lbrace \mathbf{b}\in \F_{q}^{n}:\ \mathbf{b}\cdot\mathbf{a}^{T}=0 \ \text{for any}\ \mathbf{a}\in \C \rbrace$
	\end{center} and $\cdot$ denotes the standard inner product.
The code $\C$ is called $\lambda $-constacyclic if
	$(a_{0}, a_{1},\ldots, a_{n-1}) \in\C$ implies 
	$(\lambda a_{n-1}, a_{0}, a_{1},\dots, a_{n-2})\in\C$,
	$\C$ is cyclic if $\lambda=1$ and $\C$ is negacyclic if $\lambda=-1$.
	Representing $(a_{0}, a_{1},\ldots, a_{n-1}) \in \C$ as
	$a_{0}+a_{1}x+a_{2}x^{2}+\cdots+a_{n-1}x^{n-1}\in \F_q[x]/ \langle x^{n}-\lambda\rangle$, any $\lambda $-constacyclic code 
	$\C$ of length $n$ is equivalent to
	an ideal of $\F_q[x]/ \langle x^{n}-\lambda\rangle$,
	then there is a monic $g(x)$ such that $\C=\langle g(x) \rangle$. If $g(x)$ is the smallest degree, then $g(x)$ 
	and $h(x)=\frac{x^n-\lambda}{g(x)}$ are referred to as the generator and parity-check polynomial of $\C$. Let $\C$ have $s$ nonzeros, i.e.,
	$h(x)$ has $s$ irreducible factors over $\F_q$.
	The extended linear code $\overline{\C}$ is 
	$$\overline{\C}=\{(a_0,a_1,\dots, a_{n-1},a_n):(a_0,a_1,\dots, a_{n-1})\in \C,\ \sum_{i=0}^{n}a_i=0 \} .$$
	
	In this paper, let $n,q\in Z^+$ and $\gcd (n,q)=1$. For $b,\delta\in Z$. Let $m = \ord_{n}(q)$, $\F_{q^m}^*=\langle\alpha\rangle$ and $\beta= \alpha^{\frac{q^{m}-1}{n}}$. 
	We define $m_i(x)$ as the minimal polynomial corresponding to $\beta ^{i}$ with $i\in[0,n-1]$. Define	
	$$g_{(n,1,\delta,b)}(x):=\lcm \left(m_{b}(x),m_{b+1}(x),\ldots,m_{b+\delta -2}(x)\right) \ \text{and} \ \delta \in[2,n],$$ 
	 $\lcm$ represents the lowest common multiple. 
	Let $l = \ord_{2n}(q)$, $\F_{q^l}^*=\langle\alpha\rangle$ and $\gamma= \alpha^{\frac{q^{l}-1}{2n}}$. We define $M_{1+2i}(x)$ as the minimal polynomial corresponding to $\gamma ^{1+2i}$ with $i\in[0,n-1]$. Define	
	$$g_{(n,-1,\delta,b)}(x):=\lcm \left(M_{1+2b}(x),M_{1+2(b+1)}(x),\ldots,M_{1+2(b+\delta-2)}(x)\right) \ \text{and} \ \delta \in[2,n].$$
	Let $\C_{(n,1 ,\delta ,b)}=\langle g_{(n,1,\delta,b)}(x) \rangle$ and $\C_{(n,-1 ,\delta ,b)}=\langle g_{(n,-1,\delta,b)}(x)\rangle$. 
	Then $\C_{(n,1 ,\delta ,b)}$ is a cyclic BCH code with the defining set $T=\bigcup _{i=b}^{b+\delta-2}C_{i}^{n}$, $\C_{(n,-1 ,\delta ,b)}$ is a negacyclic BCH code with the defining set $T=\bigcup _{i=b}^{b+\delta-2}C_{1+2i}^{2n}$, where 
	$C_{i}^{n}$ is the $q$-cyclotomic coset of $i$ (modulo $n$). Clearly, the dimension of $\C_{(n,\lambda ,\delta,b)}$ is $n-\left\lvert T\right\rvert $, where $\lambda\in \{1,-1\}$.
	Let $\widetilde{g}_{(n,1,\delta,b)}(x)=(x-1)g_{(n,1,\delta,b)}(x)$ and 
	$\widetilde{\C}_{(n,1 ,\delta,b)}=\langle  \widetilde{g}_{(n,1 ,\delta,b)}(x)\rangle $. 
 
	BCH codes studied for a long time and had many results in \cite{RefJ1}, \cite{RefJ3}, \cite{RefJ6}, \cite{RefJ8}, \cite{RefJ15}-\cite{RefJ21},
	\cite{RefJ22}-\cite{RefJ31}.
	In \cite{RefJ10}, \cite{RefJ27}, \cite{RefJ33}, \cite{RefJ34}, the authors examined the negacyclic codes 
	of length $n=\frac{q^{2m}-1}{2}$, $\frac{q^{m}+1}{2}$, $\frac{q^{m}-1}{4}$ and $\frac{q^{m}-1}{2\lambda}$ respectively.
	In \cite{RefJ12}, \cite{RefJ18}, \cite{RefJ21},  the authors provided some research about weight distributions. 
	There are many conclusions on the parameters in \cite{RefJ4}, \cite{RefJ5}, \cite{RefJ14}, \cite{RefJ32}.
	Next, we list some known results for BCH codes of length $n=\frac{q^m-1}{2}$.
	
	For negacyclic BCH code, Wang et al. \cite{RefJ27} presented ${\delta_i}'$ 
	(where $i=1,2,3$ and ${\delta_i}'$ be defined in Section 2) and analysed the parameters of $\C_{(n,-1,\delta,0)}$ for $\delta $ in some ranges. Sun et al. \cite{RefJ24} gave a class of ternary
	negacyclic code with parameters $[\frac{3^m-1}{2}, 2m,d\geq \frac{3^{m-1}-1}{2}]$ and its dual code with parameters $[\frac{3^m-1}{2},\frac{3^m-1}{2}-2m,5]$.
	
	For cyclic BCH code, Zhu et al. \cite{RefJ35} determined the dimension of $\C_{(n,1 ,\delta ,1)}$ with designed distance
	$2\leq \delta \leq \frac{q^{\left\lceil\frac{m+1}{2}\right\rceil }+1}{2}$, and gave ${\delta_i}$ (where $i=1,2,3$ and $\delta_i$ be defined in Section 2) and the weight distributions of $\C_{(n,1 ,\delta_1,1)}$ and $\C_{(n,1 ,\delta_2 ,1)}$.
	Ling et al.\cite{RefJ19} presented the minimum distances and weight enumerators of $\C_{(n,1 ,\delta ,1)}$ for
	$\delta=\frac{q^m-q^{m-1}}{2}-1-\frac{q^{\left\lfloor \frac{m-3}{2}\right\rfloor +i}-1}{2}(i \in [1,\left\lfloor \frac{m+11}{6}\right\rfloor ])$
	using the trace representation.
	Wang et al. \cite{RefJ26} gave the parameters of $\C_{(n,1 ,\delta ,1)}$ for 
	$\delta=\frac{q^m-q^{m-1}-q^{\left\lfloor \frac{m-3}{2}\right\rfloor+i}-1}{2}$ and $\frac{q^m-q^{m-1}-q^{\frac{3m-1}{4}}-q^{\frac{m+1}{2}}-q^{\frac{m-11}{4}+t}-1}{2}$, where $1\leq i\leq \left\lfloor \frac{m+6}{4}\right\rfloor $ and $2\leq t \leq 4$.
	
	In this paper, let $n=\frac{q^m-1}{2}$. The paper is segmented in the following sections: Section 2 gives some preliminaries.
	Section 3 presents the dimension of negacyclic codes $\C_{(n,-1 ,\delta ,0)}$ with few nonzeros and $\delta$ takes some special forms. 
	Section 4 gives the weight distribution of extended code and the parameters of dual code of cyclic BCH code.
	Section 5 serves as the conclusions of this paper.		

	\section{Preliminaries}
	In this section, we introduce some notations and results about BCH codes.
	
	\begin{itemize}
		\item $T^{\perp}$ is the defining set of $\C^{\perp}$.
		\item $CL(t)=\min\{i: i\in C_t^n\}$ is called the coset leader of $C_t^n$.
		\item $\Gamma_{n}=\{CL(t):t\in Z_n\}$ is the set of all coset leaders.
		\item $\delta _i \in \Gamma_{n}$ is the $i$-th largest.
		\item ${\delta _i}' \in \Gamma_{2n}$ is odd and is the $i$-th largest.
        \item $\left\lfloor a\right\rfloor=\max\{{a}':{a}'\leq a, {a}'\in Z \}$. 
        \item $\left\lceil a\right\rceil=\min\{{a}':{a}'\geq a, {a}'\in Z \}$. 
		\end{itemize}
Let $Tr_q^{q^r}$ be the trace mapping from $\F_{q^r}$ to $\F_q$. Let $t=\sum_{s=0}^{m-1}t_sq^s$ be the $q$-adic
expansion with $t\in Z$ and $0<t<q^m-1$. Denote $t=(t_{m-1},t_{m-2}, \ldots,t_0)$ is the sequence of $t$, 
the sequence of $tq^j$ mod $n$ is $[tq^j]_n=(t_{m-1-j},t_{m-2-j}, \ldots,t_{m-j+1},t_{m-j})$ which is the circular $j$-left-shift of $t$.

	\begin{lemma}(\cite{RefJ27})\label{lem1}
			Let $i=(i_{m-1},i_{m-2},\ldots,i_0)$ for $i\in [1,q^m-2]$. Then $i\in \Gamma_{q^m-1}$
			iff $[iq^j]_{q^m-1}\geq i$ for $j \in [0,m-1]$. 
	\end{lemma}
	\begin{lemma}(\cite{RefJ13})\label{lem2}
		Let $\C$ be a negacyclic code of length $n$ and $\gamma$ be a primitive $2n$-th root of unity in $\F_{q^m}$.
		If $\{\gamma^{1+2s}:t\leq s\leq t+\delta -2\}\subseteq T$, then $d(\C)\geq\delta $.
	\end{lemma}

	\begin{lemma}(\cite{RefJ16})\label{lem3}
		Let $m=\ord_{n}(q)$, $\F_{q^m}^*=\langle\alpha\rangle$ and $\beta= \alpha^{\frac{q^{m}-1}{n}}$. 
		Let $\C$ be a cyclic code of length $n$ over $\F_q$ and has $s$ nonzeros.
		Let $\beta^{i_1}, \beta^{i_2}, \ldots,\beta^{i_s}$ be the $s$ roots of $h(x)$, they are not mutually conjugate.
		Denote $|C_{i_t}^n|=m_t$ for $1\leq t\leq s$. Then the trace representation of $\C$
		is $$\C=\{c(a_1,a_2,\ldots,a_s):a_t\in\F_{q^{m_t}}, 1\leq t\leq s\},$$
		where $c(a_1,a_2,\ldots,a_s)=(\sum_{t=1}^{s}Tr_q^{q^{m_t}} (a_t\beta ^{-li_t}))_{l=0}^{n-1}.$
	\end{lemma}

	\begin{lemma}(\cite{RefJ7}, \cite{RefJ11})\label{lem4}
		 Let $\C$ be an $[n,k,d]$ code over $\F_q$. Then 
		\[\begin{split}
			\begin{cases}
				\sum_{i=0}^{\lfloor \frac{d-1}{2}\rfloor}(q-1)^{i} \binom{n}{i}\leq q^{n-k},\\
				\sum_{i=0}^{\frac{d-2}{2}}(q-1)^{i}\binom{n-1}{i}\leq q^{n-k-1}, &\text{if $2\mid d.$}\\  
		\end{cases} \end{split}\]
	\end{lemma}
	
	\begin{lemma}(\cite{RefJ20})\label{lem5}
		Let $m\geq 3$, $2\nmid m$ and $1 \leq i \leq q^{\frac{m+1}{2}}-1$. Then $i\in \Gamma_{2n}$ iff 
		$q\nmid i$. Moreover, $|C_i^{2n}|=m$.
	\end{lemma}
	
	\begin{lemma}(\cite{RefJ20})\label{lem6}
		Let $2\mid m$ and $1 \leq i \leq 2q^{\frac{m}{2}}-1$. Then $i\in \Gamma_{2n}$ iff 
		$q\nmid i$. Moreover, $|C_{q^\frac{m}{2}+1}^{2n}|=\frac{m}{2}$ and 
		$|C_i^{2n}|=m$ for $i\neq q^\frac{m}{2}+1$.
	\end{lemma}
	
	\begin{lemma}(\cite{RefJ27})\label{lem7}
		Let $2\leq \delta \leq q^{\frac{m}{2}}+1$ with $2\mid m$ and  $2\leq \delta \leq \frac{q^{\frac{m+1}{2}}+1}{2}$ with $2\nmid m$.
		Then $\C_{(n,-1,\delta ,0)}$ is an $[n,n-m\left\lceil \frac{(2\delta-3)(q-1)}{2q}\right\rceil,d]$ negacyclic BCH code, 
		where $$d\geq \begin{cases} \delta +1,& \text{if $q\mid (\delta-\frac{q+1}{2})$};\\
			\delta ,& \text{otherwise}.\\ \end{cases}$$
	\end{lemma}

	\section{The parameters of negacyclic BCH code of length $n=\frac{q^m-1}{2}$}\label{set3}
	Let $n=\frac{q^m-1}{2}$.
	We will give $\dim(\C_{(n,-1,\left\lceil \frac{\widetilde{\delta}+1}{2}\right\rceil,0)})$ for $\widetilde{\delta} =a\frac{q^m-1}{q-1}$, $a\frac{q^m-1}{q-1}+b\frac{q^m-1}{q^2-1}$
	and $\frac{{\delta_i}'+1}{2}$. Furthermore, 
	we discuss $\dim(\C_{(n,-1,\delta ,b)})$ and the range of $d (\C_{(n,-1,\delta,1)})$ for $b\neq 0$.
	
	\subsection{The dimension of $\C_{(n,-1,\delta ,0)}$ with four specific designed distances}\label{set3.1}
	Let $1\leq a <\frac{q-1}{2}$, $q>3$ and $m>2$ in Lemmas \ref{lem8} and \ref{lem10}.
	\begin{lemma}\label{lem8}
		 Let $\delta =\left\lceil \frac{\widetilde{\delta}+1}{2}\right\rceil $ and $\widetilde{\delta} =a\frac{q^m-1}{q-1}$.
		Then $\C_{(n,-1,\delta,0)}$ is an $[n,k,d\geq \delta]$ negacyclic BCH code, where
		$$k=\sum_{j=0}^{\left\lfloor \frac{m-1}{2}\right\rfloor}\left\lceil \frac{q-a-1}{2}\right\rceil ^{2j+1}
		\left\lfloor \frac{q-a+1}{2}\right\rfloor^{m-2j-1}\binom{m}{2j+1}.$$
		\begin{proof}
			Note that
			$$\widetilde{\delta}=(a,a,\ldots,a,a).$$	
			For code $\C_{(n,-1,\delta,0)}$, we have $T=\bigcup _{i=0}^{\delta -2}C_{1+2i}^{2n}$ and 
			$\dim(\C_{(n,-1,\delta,0)})=n-\left\lvert T\right\rvert$. 
			To determine $\dim(\C_{(n,-1,\delta,0)})$, we need to determine the number of $s\notin T$ and $s\in \Z_{2n}$ with $2\nmid s$.
			Let $A=\{s:s\in \Z_{2n},2\nmid s,s\notin \bigcup _{i=1}^{\widetilde{\delta} -1}C_{i}^{2n}\}$. 
			Hence, we can get $\dim(\C_{(n,-1,\delta,0)})=|A|$. 

			Let ${s}'=\sum_{t=0}^{m-1}s_{t}q^t$ 
			be the $q$-adic expansion. Let ${s}'\notin \bigcup _{i=1}^{\widetilde{\delta} -1}C_{i}^{2n}$ and ${s}'\in \Z_{2n}$.
			Clearly, ${s}'\notin \bigcup _{i=1}^{\widetilde{\delta}-1}C_{i}^{2n}$
			$\Leftrightarrow $ $a \leq s_t \leq q-1$ for any $0 \leq t \leq m-1$ by Theorem 3.2 of \cite{RefJ2}.
			Let $B=\{s_t :2\nmid s_t, 0 \leq t \leq m-1\}$. Since $q^t$ is odd, then $2\nmid {s}'\Leftrightarrow 2\nmid \left\lvert B\right\rvert $.\\
			$\bullet$ $\textbf{Case 1:}$ $2 \mid m$.
			
			If $2\nmid a$, then
			$|\{t:2\nmid s_t\}|=\frac{q-a}{2}$ and $|\{t:2\mid s_t\}|=\frac{q-a}{2}$ for any $s_t\in [a,q-1]$.
			Since $2\nmid {s}'\Leftrightarrow 2\nmid \left\lvert B\right\rvert $, then $|A|=\sum_{j=0}^{\frac{m}{2}-1}\binom{m}{2j+1}(\frac{q-a}{2})^m$.			
			If $2\mid a$, then 
			$|\{t:2\nmid s_t\}|=\frac{q-a-1}{2}$ and $|\{t:2\mid s_t\}|=\frac{q-a+1}{2}$ for any $s_t\in [a,q-1]$.
			Since $2\nmid {s}'\Leftrightarrow 2\nmid \left\lvert B\right\rvert $, then 
			$|A|=\sum_{j=0}^{\frac{m}{2}-1}\binom{m}{2j+1}(\frac{q-a-1}{2})^{2j+1}(\frac{q-a+1}{2})^{m-2j-1}$.\\
			$\bullet$ $\textbf{Case 2:}$ $2\nmid m$.
			
			If $2\nmid a$, since $|\{t:2\nmid s_t\}|=|\{t:2\mid s_t\}|=\frac{q-a}{2}$ and $2\nmid {s}'\Leftrightarrow 2\nmid \left\lvert B\right\rvert $, then 
			$|A|=\sum_{j=0}^{\frac{m-1}{2}}\binom{m}{2j+1}(\frac{q-a}{2})^m$. By the same way, we have $|A|=\sum_{j=0}^{\frac{m-1}{2}}\binom{m}{2j+1}(\frac{q-a-1}{2})^{2j+1}(\frac{q-a+1}{2})^{m-2j-1}$ for $2\mid a$.
			
		Summing up all the discussions above, we complete our proof.
		\end{proof}
	\end{lemma}
	
	\begin{example}
		Let $(q,m)=(5,3)$ and $a=1$, one has $n=62$, $\widetilde{\delta}=31$ and $\delta=16$. Then $\dim (\C_{(62,-1,16,0)})=32$.
		Let $(q,m)=(7,3)$ and $a=2$, one has $n=171$, $\widetilde{\delta}=114$ and $\delta=58$. Then $\dim (\C_{(171,-1,58,0)})=62$.
	\end{example}
	
	\begin{lemma}\label{lem10}
		Let $\delta =\left\lceil \frac{\widetilde{\delta}^{'}+1}{2}\right\rceil$ and $\widetilde{\delta}^{'}=aq^{m-1}-1$.
		Then $\C_{(n,-1,\delta,0)}$ is an $[n,k,d\geq \delta]$ negacyclic BCH code, where
		$$k=\sum_{j=0}^{\left\lfloor \frac{m-1}{2}\right\rfloor}\left\lceil \frac{q-a-1}{2}\right\rceil ^{2j+1}
		\left\lfloor \frac{q-a+1}{2}\right\rfloor^{m-2j-1}\binom{m}{2j+1}+\kappa m,$$ $\kappa=1$ for $2\mid a$ and $\kappa=0$ for $2\nmid a$.
		\begin{proof}
From Theorem 3.5 of \cite{RefJ2}, 
we know that $\bigcup _{i=1}^{\widetilde{\delta}^{'}}C_{i}^{2n}=\bigcup _{i=1}^{\widetilde{\delta}-1}C_{i}^{2n}$, $\widetilde{\delta}^{'}\in \Gamma_{2n}$ and $|C_{\widetilde{\delta}^{'}}^{2n}|=m$, 
where $\widetilde{\delta} =a\frac{q^m-1}{q-1}$. Similar to Lemma \ref{lem8}, we can get 
$$\begin{aligned}
	\dim(\C_{(n,-1,\delta,0)})&=|\{s:s\in \Z_{2n},2\nmid s,s\notin \bigcup _{i=1}^{\widetilde{\delta}^{'} -1}C_{i}^{2n}\}|\\
	&=|\{s:s\in \Z_{2n},2\nmid s,s\notin \bigcup _{i=1}^{\widetilde{\delta}^{'} }C_{i}^{2n}\}|+|\{s:2\nmid s,s\in C_{\widetilde{\delta}^{'}}^{2n}\}|.
	\end{aligned}$$
Since $2\mid \widetilde{\delta}^{'}$ for $2\nmid a$ and $2\nmid \widetilde{\delta}^{'}$ for $2\mid a$, then $$|\{s:2\nmid s,s\in C_{\widetilde{\delta}^{'}}^{2n}\}|=\begin{cases}
	m\ &\text{if}\ 2\mid a;\\
	0\ &\text{if}\ 2\nmid a.\end{cases}$$
Combining Lemma \ref{lem8}, we get $\dim(\C_{(n,-1,\delta,0)})$.
\end{proof} \end{lemma}
	
	We assume that $1\leq a+b \leq q-1$, $m>2$ is even and $\left\lceil \frac{q-a-2}{2}\right\rceil\geq 1$ in Lemmas \ref{lem11} and \ref{lem13}.
	
	\begin{lemma}\label{lem11}
		Let $\delta =\left\lceil \frac{\widetilde{\delta}+1}{2}\right\rceil$ and $\widetilde{\delta} =a\frac{q^m-1}{q-1}+b\frac{q^m-1}{q^2-1}$. 
		Then $\C_{(n,-1,\delta,0)}$ is an $[n,k,d\geq \delta]$ negacyclic BCH code, where
		$$k=\sum_{t=0}^{\frac{m}{2}}\frac{m}{m-t}\binom{m-t}{t} \sum_{j=0}^{t} \binom{t}{j}\left\lfloor 
		\frac{q-a-b}{2}\right\rfloor^j \left\lceil \frac{q-a-b}{2}\right\rceil^{t-j} \phi,$$
		when $2\nmid a$,
		\begin{center}
			$\phi$=$\begin{cases}
				\sum_{i=0}^{\frac{m}{2}-t}\binom{m-2t}{2i} \left\lceil \frac{q-a-2}{2}\right\rceil^{2i}\left\lfloor \frac{q-a}{2}\right\rfloor^{m-2t-2i} ,
				& \text{if $2\nmid t$ and $2\mid j$, or $2\mid t$ and $2\nmid j$ } ;\\
				\sum_{i=0}^{\frac{m}{2}-t-1}\binom{m-2t}{2i+1} \left\lceil \frac{q-a-2}{2}\right\rceil^{2i+1}\left\lfloor \frac{q-a}{2}\right\rfloor^{m-2t-2i-1} ,
				& \text{if $2\nmid t$ and $2\nmid j$, or $2\mid t$ and $2\mid j$ };\\
			\end{cases}$
		\end{center}
		when $2\mid a$,
		\begin{center}
			$\phi$=$\begin{cases}
				\sum_{i=0}^{\frac{m}{2}-t}\binom{m-2t}{2i} \left\lceil \frac{q-a-2}{2}\right\rceil^{2i}\left\lfloor \frac{q-a}{2}\right\rfloor^{m-2t-2i} ,
				& \text{if $2\nmid j$} ;\\
				\sum_{i=0}^{\frac{m}{2}-t-1}\binom{m-2t}{2i+1} \left\lceil \frac{q-a-2}{2}\right\rceil^{2i+1}\left\lfloor \frac{q-a}{2}\right\rfloor^{m-2t-2i-1} ,
				& \text{if $2\mid j$ }.\\
			\end{cases}$
		\end{center}
		\begin{proof}
			Note that 
$$\widetilde{\delta}=(a,a+b,a,a+b,\ldots,a,a+b).$$
Let $A=\{s:s\in \Z_{2n},2\nmid s,s\notin \bigcup _{i=1}^{\widetilde{\delta} -1}C_{i}^{2n}\}$. By the same way as Lemma \ref{lem8}, we have $\dim(\C_{(n,-1,\delta,0)})=|A|$. 
Let ${s}'=\sum_{u=0}^{m-1}s_{u}q^u\notin \bigcup _{i=1}^{\widetilde{\delta} -1}C_{i}^{2n}$ and ${s}'\in \Z_{2n}$. This is equivalent to $a+b \leq s_{(u-1)}\leq q-1$ if $s_u=a$ and $a \leq s_{(u-1)}\leq q-1$ if $s_u>a$ by Theorem 3.1 of \cite {RefJ22}, where “$s_{(u-1)}$” means “$s_{(u-1) \ ({\rm mod}~m)}$”.
Let $B=\{u:s_u=a, \ s_{(u-1)}\in [a+b,q-1], \ 0 \leq u \leq m-1 \} $ and $\left\lvert B\right\rvert=t$. It is clear that $0 \leq t \leq \frac{m}{2}$ 
and there exist $\frac{m}{m-t}\binom{m-t}{t}$ alternatives for $|B|=t$.
Since $2\nmid {s}'\Leftrightarrow {s}'\in A \Leftrightarrow 2\nmid (|\{u: 2\nmid s_u\}|+|\{u-1: 2\nmid s_{(u-1)}\}|+|\{i:2\nmid s_{i}\}|)$ with $u\in B$ and $i\in \Z_m\backslash \{u,(u-1) \ ({\rm mod}~m):u\in B\}$, then we will discuss the parity of $a$, $b$ and $t$.\\
$\bullet$ $\textbf{Case 1:}$ $2\nmid a,b$.
			
In this case, $s_{(u-1)}\in [a+b,q-1]$ and $s_{i}\in [a+1,q-1]$ for $i\in \Z_m\backslash \{u,(u-1) \ ({\rm mod}~m):u\in B\}$. 
Clearly, $|\{u-1: 2\nmid s_{(u-1)}\}|=\frac{q-a-b-1}{2}$ and 
$|\{u-1:2\mid s_{(u-1)}\}|=\frac{q-a-b+1}{2}$ for $u\in B$ and $s_{(u-1)}\in [a+b,q-1]$. Furthermore,
$|\{i:s_{i}> a,2\nmid s_{i}\}|=\frac{q-a-2}{2}$ and $|\{i:s_{i}> a,2\mid s_{i}\}|=\frac{q-a}{2}$ for $i\in \Z_m\backslash \{u,(u-1) \ ({\rm mod}~m):u\in B\}$.
			
If $2\nmid t$, then $|A|=\frac{m}{m-t}\binom{m-t}{t} \sum_{j=0}^{t}\binom{t}{j}(\frac{q-a-b-1}{2})^j(\frac{q-a-b+1}{2})^{t-j}\phi_1$, where
\begin{center}
	$\phi _1$=$\begin{cases}
\sum_{i=0}^{\frac{m}{2}-t}\binom{m-2t}{2i}  (\frac{q-a-2}{2})^{2i} (\frac{q-a}{2})^{m-2t-2i},  & \text{if $2\mid j$} ;\\
\sum_{i=0}^{\frac{m}{2}-t-1}\binom{m-2t}{2i+1}  (\frac{q-a-2}{2})^{2i+1} (\frac{q-a}{2})^{m-2t-2i-1}, & \text{if $2\nmid j$ }.\\
\end{cases}$ 
\end{center}
			
If $2\mid t$, then $|A|=\frac{m}{m-t}\binom{m-t}{t} \sum_{j=0}^{t}\binom{t}{j}(\frac{q-a-b-1}{2})^j(\frac{q-a-b+1}{2})^{t-j}\phi _2$, where
\begin{center}
	$\phi _2$=$\begin{cases}
\sum_{i=0}^{\frac{m}{2}-t-1}\binom{m-2t}{2i+1}  (\frac{q-a-2}{2})^{2i+1} (\frac{q-a}{2})^{m-2t-2i-1}, & \text{if $2\mid j$ };\\
\sum_{i=0}^{\frac{m}{2}-t}\binom{m-2t}{2i}  (\frac{q-a-2}{2})^{2i} (\frac{q-a}{2})^{m-2t-2i},  & \text{if $2\nmid j$} .\\
\end{cases}$ \end{center}
$\bullet$ $\textbf{Case 2:}$ $2\nmid a$ and $2\mid b$.
			
In this case, $s_{(u-1)}\in [a+b,q-1]$ for $u \in B$ and $s_{i}\in [a+1,q-1]$ for $i\in \Z_m\backslash \{u,(u-1) \ ({\rm mod}~m):u\in B\}$. 
Clearly, $|\{u-1: 2\nmid s_{(u-1)}\}|=\frac{q-a-b}{2}$ and 
$|\{u-1:2\mid s_{(u-1)}\}|=\frac{q-a-b}{2}$ for $s_{(u-1)}\in [a+b,q-1]$. Furthermore,
$|\{i:s_{i}> a,2\nmid s_{i}\}|=\frac{q-a-2}{2}$ and $|\{i:s_{i}> a,2\mid s_{i}\}|=\frac{q-a}{2}$ for $i\in \Z_m\backslash \{u,(u-1) \ ({\rm mod}~m):u\in B\}$.
							
If $2\nmid t$, then $|A|=\frac{m}{m-t}\binom{m-t}{t} \sum_{j=0}^{t}\binom{t}{j}(\frac{q-a-b}{2})^j(\frac{q-a-b}{2})^{t-j}\phi _3$, 
where $\phi _3=\phi_1$.					

If $2\mid t$, then $|A|=\frac{m}{m-t}\binom{m-t}{t} \sum_{j=0}^{t}\binom{t}{j}(\frac{q-a-b}{2})^j(\frac{q-a-b}{2})^{t-j}\phi _4$, where $\phi _4=\phi_2$.\\
$\bullet$ $\textbf{Case 3:}$ $2\mid a,b$.

In this case, $|\{u:2\nmid s_u,u\in B\}|=0$. 
By the same way as above, $|\{u-1: 2\nmid s_{(u-1)}\}|=\frac{q-a-b-1}{2}$ and 
$|\{u-1:2\mid s_{(u-1)}\}|=\frac{q-a-b+1}{2}$ for $s_{(u-1)}\in [a+b,q-1]$.  Furthermore,
$|\{i:s_{i}> a,2\nmid s_{i}\}|=\frac{q-a-1}{2}$ and $|\{i:s_{i}> a,2\mid s_{i}\}|=\frac{q-a-1}{2}$ for $i\in \Z_m\backslash \{u,(u-1) \ ({\rm mod}~m):u\in B\}$.
			
Then $|A|=\frac{m}{m-t}\binom{m-t}{t} \sum_{j=0}^{t}\binom{t}{j}(\frac{q-a-b-1}{2})^j(\frac{q-a-b+1}{2})^{t-j}\phi_5$, where 			
\begin{center}
	$\phi_5$=$\begin{cases}
\sum_{i=0}^{\frac{m}{2}-t}\binom{m-2t}{2i}  (\frac{q-a-1}{2})^{2i} (\frac{q-a-1}{2})^{m-2t-2i},  & \text{if $2\nmid j$} ;\\
\sum_{i=0}^{\frac{m}{2}-t-1}\binom{m-2t}{2i+1}  (\frac{q-a-1}{2})^{2i+1} (\frac{q-a-1}{2})^{m-2t-2i-1}, & \text{if $2\mid j$}.\\
\end{cases}$ \end{center}
$\bullet$ $\textbf{Case 4:}$ $2\mid a$ and $2\nmid b$.
			
In this case, $|\{u:2\nmid s_u,u\in B\}|=0$. 
Clearly, $|\{u-1: 2\nmid s_{(u-1)}\}|=\frac{q-a-b}{2}$ and 
$|\{u-1:2\mid s_{(u-1)}\}|=\frac{q-a-b}{2}$ for $s_{(u-1)}\in [a+b,q-1]$. Furthermore, 
$|\{i:s_{i}> a,2\nmid s_{i}\}|=\frac{q-a-1}{2}$ and $|\{i:s_{i}> a,2\mid s_{i}\}|=\frac{q-a-1}{2}$ for $i\in \Z_m\backslash \{u,(u-1) \ ({\rm mod}~m):u\in B\}$.
			
Then $|A|=\frac{m}{m-t}\binom{m-t}{t} \sum_{j=0}^{t}\binom{t}{j}(\frac{q-a-b}{2})^j(\frac{q-a-b}{2})^{t-j}\phi_6$, where $\phi _6=\phi_5$.			

	Summarizing the above, we arrive at the results.
			\end{proof}
		\end{lemma}
		\begin{example}
Let $(q,m)=(5,4)$, one has $n=312$. When $(a,b)=(1,2)$, we have $\widetilde{\delta}=208$ and $\delta =105$, then $\dim (\C_{(312,-1,105,0)})=80$.
			When $(a,b)=(2,2)$, we have $\widetilde{\delta}=364$ and $\delta =183$, then $\dim (\C_{(312,-1,183,0)})=16$.
		\end{example}
		
		\begin{lemma}\label{lem13}
			Let $\delta =\left\lceil \frac{\widetilde{\delta}^{'}+1}{2}\right\rceil$ and $\widetilde{\delta}^{'}=a(q^{m-1}+q^{m-2})+bq^{m-2}-1$.
		Then $\C_{(n,-1,\delta,0)}$ is an $[n,k,d\geq \delta]$ negacyclic BCH code, where
			$$k=\sum_{t=0}^{\frac{m}{2}}\frac{m}{m-t}\binom{m-t}{t} \sum_{j=0}^{t} \binom{t}{j}\left\lfloor 
			\frac{q-a-b}{2}\right\rfloor^j \left\lceil \frac{q-a-b}{2}\right\rceil^{t-j} \phi+\kappa m,$$ 
			 $\phi$ is given in Lemma \ref{lem11}, $\kappa=1$ for $2\mid b$ and $\kappa=0$ for $2\nmid b$.
			\begin{proof}
		From the Theorem 3.4 of \cite{RefJ22}, $r\notin \Gamma_{2n}$ for any $\widetilde{\delta}^{'}<r<\widetilde{\delta}$, where $\widetilde{\delta}=a\frac{q^m-1}{q-1}+b\frac{q^m-1}{q^2-1}$, i.e.,
		$\bigcup _{i=1}^{\widetilde{\delta}^{'}}C_{i}^{2n}=\bigcup _{i=1}^{\widetilde{\delta}-1}C_{i}^{2n}$.
		Similar to Lemma \ref{lem11}, we can get 
				$$\begin{aligned}
					\dim(\C_{(n,-1,\delta,0)})&=|\{s:s\in \Z_{2n},2\nmid s,s\notin \bigcup _{i=1}^{\widetilde{\delta}^{'} }C_{i}^{2n}\}|+ |\{s:2\nmid s,s\in C_{\widetilde{\delta}^{'}}^{2n}\}|.
				\end{aligned}$$
		Note that $ \widetilde{\delta}^{'}\in \Gamma_{2n}$ and $|C_{\widetilde{\delta}^{'}}^{2n}|=m$, combining $2\mid \widetilde{\delta}^{'}$ for $2\nmid b$ and $2\nmid \widetilde{\delta}^{'}$ for $2\mid b$, we have $$|\{s:2\nmid s,s\in C_{\widetilde{\delta}^{'}}^{2n}\}|=\begin{cases}
				m\ &\text{if}\ 2\mid b;\\
				0\ &\text{if}\ 2\nmid b.
			\end{cases}$$
		Hence, we get $\dim(\C_{(n,-1,\delta,0)})$.
		\end{proof} \end{lemma}

		We immediately get the following lemma for $a=0$ and $1\leq b\leq q-1$ in Lemma \ref{lem11}.
		\begin{lemma}\label{lem14}
		Let $\delta =\left\lceil \frac{\delta^{''}+1}{2} \right\rceil $ and $\delta^{''}=b \frac{q^m-1}{q^2-1}$. 
			Then $\C_{(n,-1,\delta,0)}$ is an $[n,k,d\geq \delta]$ negacyclic BCH code, where
			$$k=\sum_{t=0}^{\frac{m}{2}}\frac{m}{m-t}\binom{m-t}{t} \sum_{j=0}^{t} \binom{t}{j}\left\lfloor 
			\frac{q-b}{2}\right\rfloor^j \left\lceil \frac{q-b}{2}\right\rceil^{t-j} \phi,$$ where 
			\begin{center}
				$\phi$=$\begin{cases}
					\sum_{i=0}^{\frac{m}{2}-t}\binom{m-2t}{2i} \left\lceil \frac{q-2}{2}\right\rceil^{2i}\left\lfloor \frac{q}{2}\right\rfloor^{m-2t-2i} ,
					& \text{if $2\nmid j$} ;\\
					\sum_{i=0}^{\frac{m}{2}-t-1}\binom{m-2t}{2i+1} \left\lceil \frac{q-2}{2}\right\rceil^{2i+1}\left\lfloor \frac{q}{2}\right\rfloor^{m-2t-2i-1} ,
					& \text{if $2\mid j$ }.\\
				\end{cases}$
			\end{center}
	\end{lemma}
	
	\subsection{The dimensions of some codes with few nonzeros }\label{set3.2}
		\begin{lemma}\label{lem15}  (\cite{RefJ27})
			Let $n=\frac{q^m-1}{2}$. Then  
		$${\delta _1}'=q^m-q^{m-1}-1, \
		{\delta _2}'=q^m-q^{m-1}-q^{\left\lfloor \frac{2m-1}{3}\right\rfloor}-q^{\left\lfloor \frac{m-1}{3}\right\rfloor}-1$$
		and $|C_{{\delta _1'}}^{2n}|=m$, $|C_{{\delta _2}'}^{2n}|=\begin{cases}m, & \text{if $3\nmid m$},\\
			\frac{m}{3}, & \text{if $3\mid m$}\\\end{cases}$.
		If $q^m\geq 25$, then $${\delta _3}'=\begin{cases} q^m-q^{m-1}-q^{\left\lceil \frac{2m-1}{3}\right\rceil }-q^{\left\lfloor \frac{m-1}{3}\right\rfloor}-1,
			& \text{if $3\nmid (m+1)$},\\ q^m-q^{m-1}-q^{\frac{2m-1}{3}}-q^{\frac{m+1}{3}}-1, & \text{if $3\mid (m+1)$}\\
		\end{cases}$$and $|C_{{\delta _3}'}^{2n}|=m$.
		\end{lemma}
		\begin{lemma}\label{lem16}
			Let $m\equiv 2\ ({\rm mod}~3)$, $m\geq 8$ and $2 \leq i \leq 7$. Then 
			$${\delta_i}'=q^m-q^{m-1}-q^{\frac{2m-1}{3}+\varepsilon_1}-q^{\frac{m+1}{3}+\zeta_1}-1,$$
		where
			 $$\varepsilon_1=\begin{cases}0,& \text{if $2\leq i \leq 3$};\\
				1, &\text{if $4\leq i \leq 7$},\\\end{cases}  \quad \zeta_1=\begin{cases}-1,& \text{if $i=2$};\\
				1, &\text{if $i=3$};\\i-6, &\text{if $4\leq i \leq 7$}.\\\end{cases}$$
			And $\left\lvert C_{{\delta_i}'}^{2n}\right\rvert =m$ for $4\leq i \leq 7$.
		\end{lemma}
		
		\begin{proof}
			When $q\geq 3$ and $q^m\geq 25$, ${\delta _1}'$, ${\delta _2}'$ and ${\delta _3}'$ were determined in \cite{RefJ27}.
			Now we first prove that ${\delta_i}'=q^m-q^{m-1}-q^{\frac{2m-1}{3}+1}-q^{\frac{m+1}{3}+i-6}-1$ is an odd coset leader for $4\leq i \leq 7$, i.e., $q^t{\delta_i}'\geq {\delta_i}'$ for any $0\leq t \leq m-1$.
			Let  $Q_j$ be $\underbrace{q-1,\ldots,q-1}_{j}$. Clearly,
			$${\delta_i}'=(q-2,Q_{\frac{m+1}{3}-3},q-2,Q_{\frac{m-2}{3}+6-i},q-2,Q_{\frac{m+1}{3}+i-6}).$$
			
			$(1)$ If $0 \leq t \leq \frac{m+1}{3}-3$, it's obvious that $q^t{\delta_i}'> {\delta_i}'$.
			
			$(2)$ If $t =\frac{m+1}{3}-2$, then $$q^t{\delta_i}'=(q-2,Q_{\frac{m-2}{3}+6-i},q-2,Q_{\frac{m+1}{3}+i-6},q-2,Q_{\frac{m+1}{3}-3}).$$
			Note that $\frac{m-2}{3}+6-i >\frac{m+1}{3}-3 $, thus $q^t{\delta_i}'>{\delta_i}'$.
			
			$(3)$ If $ \frac{m+1}{3}-1\leq t \leq \frac{2m-1}{3}+4-i$, we have
			$$[q^t{\delta_i}']_{2n}=q^m-q^{1+t-\frac{m+1}{3}}-q^{\frac{m+1}{3}+i-6+t}-q^{t-1}-1=[q^t{\delta_i}']_{n}>{\delta_i}',$$
			since $\frac{m+1}{3}+i-6+t<m-1$ and $t-1<\frac{2m-1}{3}+1$.
			
			$(4)$ If $  t =\frac{2m-1}{3}+5-i$, then
			$$q^t{\delta_i}'=(q-2,Q_{\frac{m+1}{3}+i-6},q-2,Q_{\frac{m+1}{3}-3},q-2,Q_{\frac{m-2}{3}+6-i}),$$
			since $\frac{m+1}{3}+i-6>\frac{m+1}{3}-3$, thus $q^t{\delta_i}'>{\delta_i}'$.
			
			$(5)$ If $\frac{2m-1}{3}+6-i\leq t \leq m-1$, then
			$$[q^t{\delta_i}']_{2n}=q^m-q^{1+t-\frac{m+1}{3}}-q^{i-6+t-\frac{2m-1}{3}}-q^{t-1}-1=[q^t{\delta_i}']_{n}>{\delta_i}',$$
			by $t<m$, $1+t-\frac{m+1}{3}<\frac{2m-1}{3}+1$ and $i-6+t-\frac{2m-1}{3}<\frac{m+1}{3}+i-6$.
			
			Hence, we have ${\delta_i}'\in \Gamma _{2n}$ and $2\nmid {\delta_i}'$ for $4\leq i \leq 7$. 
			Since $q^t{\delta_i}'>{\delta_i}'$ for any $0\leq t \leq m-1$,
			then $\left\lvert C_{{\delta_i}'}^{2n}\right\rvert =m$.
			Next we prove that no integer ${\delta}'\in[{\delta}' _{i+1}+1,{\delta_i}'-1]$ satisfies ${\delta}'\in \Gamma _{2n}$ and $2\nmid {\delta}'$.
			It is easy seen that 
			$${\delta}'_{i+1}=(q-2,Q_{\frac{m+1}{3}-3},q-2,Q_{\frac{m-2}{3}+5-i},q-2,Q_{\frac{m+1}{3}+i-5}).$$
				Let $\delta ^{'} \in [{\delta }'_{i+1}+1,{\delta_i}'-1]$.
				According to the Lemma \ref{lem1}. 
	      Then the form of $\delta$ is
			$$\delta^{'}=(q-2,Q_{\frac{m+1}{3}-3},q-2,Q_{\frac{m-2}{3}+6-i},\delta^{''}_{\frac{m+1}{3}+i-6},\delta^{''}_{\frac{m+1}{3}+i-7},\ldots,\delta^{''}_0).$$
			Supposing $\delta ^{'}$ is an odd coset leader, then $q-2 \leq\delta^{''}_t\leq q-1$ for $0\leq t\leq \frac{m+1}{3}+i-7$
			and $\delta^{''}_{\frac{m+1}{3}+i-6}=q-2$.
			Let $A=\{\delta^{''}_t:\delta^{''}_t=q-2, 0\leq t\leq \frac{m+1}{3}+i-7\}$.
			If $|A|=0$, then $\delta ^{'}={\delta_i}'$, this contradicts $\delta ^{'} \in [{\delta }'_{i+1}+1,{\delta_i}'-1]$. 
			If $|A|=1$, then $2\mid \delta ^{'}$, that is impossible.
			If $|A|=2$, then 
			$$(\delta^{''}_{\frac{m+1}{3}+i-7},\ldots,\delta^{''}_0)=(Q_{a_1},q-2,Q_{a_2},q-2,Q_{a_3}),$$
			where $\sum_{j= 1}^{3}a_j=\frac{m+1}{3}+i-8$ and $a_j\geq0$ for $4\leq i\leq 7$ and $1\leq j\leq3$.
			There must exist an $a_j$ that does not satisfy $a_j\geq \frac{m+1}{3}-3$ for $1\leq j\leq3$,
			i.e., $\delta ^{'}\notin \Gamma _{2n}$. If $|A|>2$, then $\delta ^{'}$ does not satisfy ${\delta}'\in \Gamma _{2n}$ and $2\nmid {\delta}'$.
		\end{proof}
	
		\begin{theorem}\label{th17}
			Let $m\equiv 2\ ({\rm mod}~3)$, $m\geq 8$. Then $\dim(\C_{(n,-1,\frac{{\delta_i}'+1}{2},0)})=im$, where ${\delta_i}'$ is given in Lemmas \ref{lem15} and \ref{lem16}. 
		\begin{proof}
		Using lemmas \ref{lem15} and \ref{lem16}, we can directly prove this conclusion.
		\end{proof}
		\end{theorem}
		
	Using the same approach as above, the following results can be established.
		
		\begin{lemma}\label{lem18}
			Let $m\equiv 0\ ({\rm mod}~3)$, $m\geq 9$ and $2 \leq i \leq 10$. Then 
			$${\delta_i}'=q^m-q^{m-1}-q^{\frac{2m}{3}+\varepsilon_2}-q^{\frac{m}{3}+\zeta_2}-1,$$  
			 where
			$$\varepsilon_2=\begin{cases}-1,&\text{if $i=2$};\\
				0,&\text{if $3\leq i \leq 5$} ;\\1,&\text{if $6\leq i \leq 10$},\\\end{cases} \quad \zeta_2=\begin{cases}-1,&\text{if $2\leq i \leq 3$};\\
				i-4, &\text{if $4 \leq i \leq 5$};\\i-8, &\text{if $6\leq i \leq 10$}.\\\end{cases}$$
			And $\left\lvert C_{{\delta_i}'}^{2n}\right\rvert =m$  for $4\leq i \leq 10$.
		\end{lemma}
		
		\begin{theorem}\label{th19}
			Let $m\equiv 0\ ({\rm mod}~3)$, $m\geq 9$. Then 
			$\dim(\C_{(n,-1,\frac{{\delta_i}'+1}{2},0)})=
			(i-\frac{2}{3})m $, where ${\delta_i}'$ is given in Lemma \ref{lem18} with $2\leq i \leq 10$.
		\end{theorem}
		
		\begin{lemma}\label{lem20}
			Let $m\equiv 1\ ({\rm mod}~3)$, $m\geq 7$ and $2 \leq i \leq 6$. Then 
			$${\delta_i}'=q^m-q^{m-1}-q^{\frac{2m+1}{3}+\varepsilon_3}-q^{\frac{m-1}{3}+\zeta_3}-1,$$ 
			where $$\varepsilon_3=\begin{cases}-1,&\text{if $i=2$};\\
				0,&\text{if $3\leq i \leq 6$} ,\\\end{cases} \quad \zeta_3=\begin{cases}0,&\text{if $i=2$};\\
				i-4 ,&\text{if $3 \leq i \leq 6$}.\\\end{cases}$$
			And $\left\lvert C_{{\delta_i}'}^{2n}\right\rvert =m$  for $4\leq i \leq 6$.
		\end{lemma}
		
		\begin{theorem}\label{th21}
			Let $m\equiv 1\ ({\rm mod}~3)$, $m\geq 7$. Then $\dim(\C_{(n,-1,\frac{{\delta_i}'+1}{2},0)}) =im$,
			 where ${\delta_i}'$ is given in Lemma \ref{lem20} with $2 \leq i \leq 6$.

		\end{theorem}
			
		We discuss the parameters of negacyclic code ${\C}'$ with length $n=\frac{q^m-1}{2}$ and check polynomial $\lcm (M_1(x),M_{{\delta_1}'}(x))$. 
				\begin{theorem}\label{th22}
			Let $n=\frac{q^m-1}{2}$ and $m\geq3$. Then ${\C}'$ is an $[n,2m,d({\C}')\geq \frac{(q-2)q^{m-1}-1}{2}]$ code
	and ${\C}'^{\perp}$ is an $[n,n-2m,3\leq d({\C}'^{\perp})\leq 5]$ dual code.
\end{theorem}

\begin{proof}
	 Since $1\in \Gamma_{2n}$ and ${\delta_1}'\in \Gamma_{2n}$, then  ${\delta_1}'\notin C_1^{2n}\Rightarrow\lcm (M_1(x),M_{{\delta_1}'}(x))=M_1(x)M_{{\delta_1}'}(x)$.
	From Lemma \ref{lem15}, we have $|C_1^{2n}|=|C_{{\delta_1}'}^{2n}|=m$ $\Rightarrow $ $\dim({\C}')=2m$ and $\dim({\C}'^{\perp})=n-2m$.
	
	Let $$H=\{1+2i:\frac{q^{m-1}+1}{2}\leq i\leq \frac{q^m-q^{m-1}-4}{2}\}.$$
	It is easy to get $H\subset T=\{1+2s:s\in[0,n-1]\}\setminus  (C_1^{2n}\cup C_{{\delta_1}'}^{2n})$, and $\beta ^i$ is a root of ${\C}'$ for any 
	$i\in H$. Hence, $d({\C}')\geq \frac{(q-2)q^{m-1}-1}{2}$ by Lemma \ref{lem2}.

Clearly, $-1,1\in C_1^{2n}\cup T^{\perp}= C_{{\delta_1}'}^{2n}$, then $d({\C}'^{\perp})\geq 3$ by Lemma \ref{lem2}. 
If $d({\C}'^{\perp})= 7$, by Lemma \ref{lem4}, $\sum_{i=0}^{3}(q-1)^{i} \binom{n}{i}> q^{2m}$ ,
this is infeasible, i.e., $d({\C}'^{\perp})\leq 6$. If  $d({\C}'^{\perp})= 6$, 
we have $\sum_{i=0}^{2}(q-1)^{i}\binom{n-1}{i}> q^{2m-1}$, this is infeasible, i.e., $3\leq d({\C}'^{\perp})\leq 5$. Thus we complete the proof.
\end{proof}

	\subsection{The dimension of $\C_{(n,-1,\delta,b)}$ with $b\neq 0 $}\label{set3.3}
		Let $CL(a_1,a_2)=\{1+2i\in [a_1,a_2]:\text{$1+2i\in \Gamma _{2n}$}\}$, 
we will give some results about $\C_{(n,-1,\delta,b)}$ for $b\geq1$.
		
\begin{theorem}\label{th23}
	Let $m\geq 2$, $1\leq 1+2(b+\delta -2)\leq q^{\frac{m+1}{2}}-1$ for $2\nmid m$ and $1\leq 1+2(b+\delta -2)\leq 2q^{\frac{m}{2}}-1$ for $2\mid m$.
	Then $\C_{(n,-1,\delta,b)}=\C_{(n,-1,b+\delta,0)}$ for  $1+2b\leq \left\lfloor \frac{b+\delta -2}{q}\right\rfloor $ 
	and $$\dim(\C_{(n,-1,\delta,b)})=n-m\left\lceil \frac{(2\delta+2b-3)(q-1)}{2q}\right\rceil.$$
	\begin{proof}
		We have $T=\bigcup _{i=b}^{b+\delta-2}C_{1+2i}^{2n}$ for code $\C_{(n,-1,\delta,b)}$. 
		If $2\nmid m$, then $|C_{1+2i}^{2n}|=m$ for $1\leq 1+2i\leq q^{\frac{m+1}{2}}-1$ from Lemma \ref{lem5}.
		If $2\mid m$, since $q^{\frac{m}{2}}+1$ is even, then $|C_{1+2i}^{2n}|=m$ for $1\leq 1+2i\leq 2q^{\frac{m}{2}}-1$ from Lemma \ref{lem6}.
		Hence, for any $m$ we have $$\dim(\C_{(n,-1,\delta,b)})=n-|\bigcup _{i=b}^{b+\delta-2}C_{1+2i}^{2n}|=n-m|CL(1+2b,1+2(b+\delta -2))|.$$
		Next we prove that there exists $t\in Z$ such that $1+2b\leq (1+2a)q^t \leq 1+2(b+\delta -2)$ for any $1+2a$ with $1\leq 1+2a\leq 1+2b$.
		
		Let $t\in Z$ satisfy $(1+2a)q^{t-1} <1+2b\leq (1+2a)q^t$, since $1+2a \leq1+2b$.
		It follows that $$1+2b \leq (1+2a)q^t < (1+2b)q\leq \left\lfloor \frac{b+\delta -2}{q}\right\rfloor q\leq b+\delta -2.$$
		This means 
		$CL(1,1+2(b+\delta -2))=CL(1+2b,1+2(b+\delta -2))$, then $\C_{(n,-1,\delta,b)}=\C_{(n,-1,b+\delta,0)}$.
		
		According to Lemma \ref{lem7}, we can get $|CL(1,1+2(b+\delta -2))|=\left\lceil \frac{(2\delta+2b-3)(q-1)}{2q}\right\rceil$, where $1\leq1+2(b+\delta -2)\leq q^{\frac{m+1}{2}}-1$ for $2\nmid m$ and $1\leq1+2(b+\delta -2)\leq 2q^{\frac{m}{2}}-1$ for $2\mid m$.
		Thus we complete the proof.
	\end{proof}
\end{theorem}
		
		Next we give a theorem about the bound of $d(\C_{(n,-1,\delta,1)})$.
		\begin{theorem}\label{th24}
			Let $(n,q)=1$, $q-1\mid n$, $k \in[1,q-1]$, $\delta _a\in \Z^+$ and $\delta _a\mid \frac{n}{q-1}$.
			Then 
			
			$(1)$ If $2\nmid k$ and $\delta=\frac{k\delta_a-1}{2}$, 
	then $\frac{k\delta _a-1}{2} \leq  d(\C_{(n,-1,\delta,1)})\leq \frac{k+1}{2}\delta _a$.
	
	$(2)$ If $2\mid k$ and $\delta =\frac{k\delta_a}{2}$, 
	then $\frac{k\delta _a}{2} \leq d(\C_{(n,-1,\delta,1)}) \leq (\frac{k}{2}+1)\delta _a$.
			\begin{proof}
				Define $\beta $ to be the $2n$-th root of unity in $\F_{q^l}$ that is primitive, where $l=\ord_{2n}({q})$.
	Clearly, $T=\bigcup _{i=1}^{\delta -1}C_{1+2i}^{2n}$ for code $\C_{(n,-1,\delta,1)}$.
				
	$(1)$ If $2\nmid k$ and $2\delta +1=k\delta_a$, we denote 
	$$f(x)=\frac{x^n+1}{x^{\frac{n}{\delta_a}}+1}\times\prod _{t=1}^{\frac{k-1}{2}} (x^{\frac{n}{\delta _a(q-1)}}-\beta ^{\frac{n(2t-1)}{q-1}}).$$
	Since $g_{(n,-1,\delta,1)}(\beta ^{1+2i})=0$ for $i \in [1,\delta -1]$, then $g_{(n,-1,\delta,1)}(\beta ^{(2t-1)\delta _a})=0$ for $t \in [1,\frac{k-1}{2}]$.
	We can get $f(\beta ^{1+2i})=0$ for $i \in [1,\delta -1]$.  
	Then $f(x)\in \C_{(n,-1,\delta,1)}$.
				
	We know that polynomial $\frac{x^n+1}{x^{\frac{n}{\delta_a}}+1}$ has $\delta _a$ terms in
	expansion and polynomial $\prod _{t=1}^{\frac{k-1}{2}} (x^{\frac{n}{\delta _a(q-1)}}-\beta ^{\frac{n(2t-1)}{q-1}})$ has at most $\frac{k+1}{2}$ terms.
 Hence, $w_H(f)\leq \frac{k+1}{2}\delta _a$.
	We know that $d(\C_{(n,-1,\delta,1)})\geq \delta=\frac{k\delta_a-1}{2}$ from Lemma \ref{lem2}.
	Hence, $\frac{k\delta _a-1}{2} \leq d(\C_{(n,-1,\delta,1)}) \leq \frac{k+1}{2}\delta _a$.
				
$(2)$ If $2\mid k$ and $2\delta =k\delta_a$, denote 
$$f(x)=\frac{x^n+1}{x^{\frac{n}{\delta_a}}+1}\times\prod _{t=1}^{\frac{k}{2}} (x^{\frac{n}{\delta _a(q-1)}}-\beta ^{\frac{n(2t-1)}{q-1}}).$$
Similar to $(1)$, then $f(x)\in \C_{(n,-1,\delta,1)}$ and 
$w_H(f)\leq (\frac{k}{2}+1)\delta _a$.
We know that $d(\C_{(n,-1,\delta,1)})\geq \delta=\frac{k\delta_a}{2}$ from Lemma \ref{lem2}.
Hence, $\frac{k\delta _a}{2} \leq d(\C_{(n,-1,\delta,1)}) \leq (\frac{k}{2}+1)\delta _a$.
			\end{proof}
		\end{theorem}
		
		\begin{example} 
			Let $(q,m)=(3,4)$, $n=\frac{q^m-1}{2}$ , $k=1$ and $\delta_a=10$. Then $\delta=5$ and the code $\C_{(40,-1,5,1)}$
			has parameters $[40,28,6]$, this is a best known linear code for these parameters in the Database \cite{RefJ9}.
		\end{example}
		
		\section{The parameter of cyclic BCH code of length $n=\frac{q^m-1}{2}$}\label{set4}
		Let $n=\frac{q^m-1}{2}$ with $m\geq 2$. 
		We will study the extended code and dual code of cyclic code.

\subsection{Extended code $\overline{\C}$ and dual code $\C^{\perp}$}\label{4.1}
	
For $\delta_2\leq \delta \leq \delta_1$, we discuss the weight distribution of $\overline{\C}_{(n,1,\delta,1)}$ and the parameters 
of $\C^{\perp}_{(n,1,\delta,1)}$.
\begin{lemma} \label{lem26} \cite{RefJ35}
	Let $n=\frac{q^m-1}{2}$. Then 
	$$\delta_1=\frac{q^m-1-q^{m-1}-q^{\left\lfloor \frac{m-1}{2} \right\rfloor }}{2} \quad \text{and} \quad
	\delta_2=\frac{q^m-1-q^{m-1}-q^{\left\lfloor \frac{m+1}{2} \right\rfloor }}{2}.$$
	Moreover, $|C_{\delta_1}^n|=\begin{cases} m,&\text{if $2\nmid m;$}\\
	\frac{m}{2},&\text{if $2\mid m,$}\end{cases}$ and $|C_{\delta_2}^n|=m$.
		\end{lemma}	

\begin{theorem}\label{th27}
	Let $2\mid m$ and $\delta _2+1\leq \delta \leq \delta _1$. Then $\overline{\C}_{(n,1,\delta,1)}$ is an 
	$[\frac{q^m-1}{2}+1,\frac{m}{2}+1,\delta_1+1]$ extended code and $\C^{\perp}_{(n,1,\delta,1)}$ is an 
	$[\frac{q^{m}-1}{2},\frac{q^{m}-1}{2}-\frac{m}{2}-1,d]$ dual code, where $$d=\begin{cases}
		3,&\text{if $q=3$ and $m=2;$}\\2,&\text{others}.\\
	\end{cases}$$
	Moreover, Table \ref{tab:1} gives the weight distribution of $\overline{\C}_{(n,1,\delta,1)}$.	
\end{theorem}
\begin{proof}
	Since 1 and  $\beta ^{\delta_1}$ are zeroes of parity-check polynomial of $\C$ and are non-conjugate roots of its.
	From Lemmas \ref{lem3} and \ref{lem26}, we have 
	$$\begin{aligned}
		\C_{(n,1,\delta,1)}=\{(Tr_q^{q^{\frac{m}{2}}}(a\beta^{-\delta_{1}i})+b)_{i=0}^{n-1}: a\in \F_{q^{\frac{m}{2}}},b\in \F_{q}\}\\ \end{aligned},$$ 
		 its weight distribution is equal to
		\begin{equation}\label{eq1}
\{\bm c(a,b)=(Tr_q^{q^{\frac{m}{2}}}(a\alpha^{(q^{\frac{m}{2}}+1)i})+b)_{i=0}^{n-1}:a\in \F_{q^{\frac{m}{2}}},b\in \F_{q}\}.\end{equation}
Note that $\widetilde{\C}_{(n,1 ,\delta,1)}=\langle \widetilde{g}_{(n,1,\delta,1)}(x)\rangle$, where $ \widetilde{g}_{(n,1,\delta,1)}(x)=(x-1)g_{(n,1,\delta,1)}(x)$ and $\widetilde{g}_{(n,1,\delta,1)}(1)=0$, i.e., 
	$$\sum_{i=0}^{n-1}Tr_q^{q^{\frac{m}{2}}}(a\alpha^{(q^{\frac{m}{2}}+1)i})=0.$$
	 
	We have $$\sum_{i=0}^{n-1}(Tr_q^{q^{\frac{m}{2}}}(a\alpha^{(q^{\frac{m}{2}}+1)i})+b)
	=\sum_{i=0}^{n-1}Tr_q^{q^{\frac{m}{2}}}(a\alpha^{(q^{\frac{m}{2}}+1)i})+nb=nb$$ for any codeword $\bm c(a,b)$ of (\ref{eq1}).
	Since $b\in \F_q$, then $nb=(\frac{q^{m-1}-1}{2}q+\frac{q-1}{2})b=\frac{q-1}{2}b$. 
	By the definition of extended code, 
	we know $\overline{\bm c}(a,b)=(\bm c(a,b), -\frac{q-1}{2}b)$ is the codeword of $\overline{\C}_{(n,1,\delta,1)}$,
	where $a \in \F_{q^\frac{m}{2}}$ and $b \in \F_q$. Clearly, $wt(\overline{\bm c}(a,b))=wt(\bm c(a,b))+\varepsilon $
	where $$\varepsilon =\begin{cases}0,&\text{if $b=0;$}\\1,&\text{if $b\neq 0.$}\end{cases}$$
	Thus, we give the weight distribution in Table \ref{tab:1}.

\begin{table}
	\caption{}
	\label{tab:1}
	\begin{center}
		\begin{tabular}{cc}
			\hline
			Weight   & Frequency \\
			\hline
			0&1\\
			$\frac{q^m-q^{m-1}-q^{\frac{m}{2}-1}+1}{2}$ & $(q-1)(q^{\frac{m}{2}}-1)$\\
			$\frac{(q-1)(q^{m-1}+q^{\frac{m}{2}-1})}{2}$ & $q^{\frac{m}{2}}-1$   \\
			$\frac{q^m+1}{2}$ & $q-1$\\
			\hline
		\end{tabular}
	\end{center}
\end{table}	
	For code $\C^{\perp}_{(n,1,\delta,1)}$ with $\delta_2+1\leq \delta \leq \delta_1$, we have 
$T^{\perp}=(\Z_{n}\backslash T)^{-1}=C_{0}^n\cup C_{CL(n-\delta_1)}^n$. Then $d(\C^{\perp}_{(n,1,\delta,1)})\geq2$. 
			
Since $n-\delta_1=\frac{q^{m-1}+q^{\frac{m}{2}-1}}{2}$, then $CL(n-\delta_1)=\frac{q^{\frac{m}{2}}+1}{2}$ and 
$$C_{\frac{q^{\frac{m}{2}}+1}{2}}^n=\{\frac{(1+q^{\frac{m}{2}})q^i}{2}:0\leq i\leq \frac{m}{2}-1\}.$$
Thus $|C_{\frac{q^{\frac{m}{2}}+1}{2}}^n|=\frac{m}{2}$, and $\dim(\C^{\perp}_{(n,1,\delta,1)})=n-\frac{m}{2}-1$.
			
For $q>3$, if $d(\C^{\perp}_{(n,1,\delta,1)})=3$, we have $$\sum_{i=0}^{1}(q-1)^{i}\binom{n}{i}> q^{\frac{m}{2}+1}$$ by Lemma \ref{lem4}, this is infeasible. 
Hence, $ d({\C^{\perp}_{(n,1,\delta,1)}})=2$. 
For the case of $ q=3$, we can get the result as above.
\end{proof}
\begin{theorem}\label{th28}
	Let $2\nmid m$ and $\delta _2+1\leq \delta \leq \delta _1$. Then $\overline{\C}_{(n,1,\delta,1)}$ is an 
	$[\frac{q^m-1}{2}+1,m+1,\delta_1+1]$ extended code and $\C^{\perp}_{(n,1,\delta,1)}$ is an $[\frac{q^{m}-1}{2},\frac{q^{m}-1}{2}-m-1,d]$ dual code,
	where $$\begin{cases}3 \leq d\leq4,&\text{if $q=3$};\\
		d=3,&\text{if $q>3$} .\\\end{cases}$$
	Moreover, Table \ref{tab:2} gives the weight distribution of $\overline{\C}_{(n,1,\delta,1)}$.	
\end{theorem}
\begin{proof}
	If $2\nmid m$, using the same approach as Theorem \ref{th27} for code $\C_{(n,1,\delta,1)}$,
	its weight distribution is equal to
	\begin{equation}\label{eq2}
		\{\bm c(a,b)=(Tr_q^{q^{m}}(a\alpha^{(q^{\frac{m+1}{2}}+1)i})+b)_{i=0}^{n-1}:a\in \F_{q^{m}},b\in \F_{q}\}\end{equation}
	 and
$$\sum_{i=0}^{n-1}Tr_q^{q^{m}}(a\alpha^{(q^{\frac{m+1}{2}}+1)i})=0.$$
We have $$\sum_{i=0}^{n-1}(Tr_q^{q^{m}}(a\alpha^{(q^{\frac{m+1}{2}}+1)i})+b)=\sum_{i=0}^{n-1}Tr_q^{q^{m}}(a\alpha^{(q^{\frac{m+1}{2}}+1)i})+nb=\frac{q-1}{2}b$$
		for any codeword $\bm c(a,b)$ of $(2)$. Similar to Theorem \ref{th27}, we directly give the weight distribution in Table \ref{tab:2}.

	\begin{table}
		\caption{}
		\label{tab:2}
		\begin{center}
			\begin{tabular}{cc}
				\hline
				Weight  & Frequency \\
				\hline
				0&1\\
				$\frac{q^m-q^{m-1}-q^{\frac{m-1}{2}}+1}{2}$ & $\frac{(q-1)(q^m-1)}{2}$\\
				$\frac{q^m-q^{m-1}}{2}$ & $q^m-1$\\
				$\frac{q^m-q^{m-1}+q^{\frac{m-1}{2}}+1}{2}$ & $\frac{(q-1)(q^m-1)}{2}$\\
				$\frac{q^m+1}{2}$ & $q-1$\\
				\hline
			\end{tabular}
		\end{center}
	\end{table}

	For code $\C^{\perp}_{(n,1,\delta,1)}$ with $\delta_2+1\leq \delta \leq \delta_1$, we have $T^{\perp}=(\Z_{n}\backslash T)^{-1}=C_{0}^n\cup C_{CL(n-\delta_1)}^n$. 
	Thus $d(\C^{\perp}_{(n,1,\delta,1)})\geq2$. 
			
Since $n-\delta_1=\frac{q^{m-1}+q^{\frac{m-1}{2}}}{2}$,
then $CL(n-\delta_1)=\frac{q^{\frac{m-1}{2}}+1}{2}$ and 
$$C_{\frac{q^{\frac{m-1}{2}}+1}{2}}^n=\{\frac{q^{l}+q^{\frac{m-1}{2}+l}}{2}: 1\leq l \leq \frac{m-1}{2}\}\cup\{\frac{q^{l}+q^{l-\frac{m+1}{2}}}{2}: \frac{m+1}{2}\leq l \leq m\}.$$
Set $m_{1}=\frac{q^{\frac{m-1}{2}}+1}{2}$, thus $|C_{m_1}^n|=m$. Then $\dim(\C^{\perp}_{(n,1,\delta,1)})=n-m-1$.
			
Suppose $d(\C^{\perp}_{(n,1,\delta,1)})=2$, which means that there exist $a(x)=a_0+a_1x+a_2x^2+\dots+a_ix^i+\dots +a_jx^j+\dots+a_{n-1}x^{n-1}\in \C^{\perp}_{(n,1,\delta,1)}$ and 
$b(x)=b_0+b_1x+b_2x^2+\dots+b_ix^i+\dots +b_jx^j+\dots+b_{n-1}x^{n-1}\in \C^{\perp}_{(n,1,\delta,1)}$, where $a_t=b_t$ for any 
$t\in[0,n-1]\setminus \{i,j\}$ and $0\leq i<j\leq n-1$.
Since $a(1)=b(1)=0$ and $a(\beta^{m_1})=b(\beta^{m_1})=0$,
i.e., $$\begin{cases} (a_i-b_i)+(a_j-b_j)=0, \\ (a_i-b_i)\beta^{im_1}+(a_j-b_j)\beta^{jm_1}=0,\end{cases}$$
which equals to 
	\begin{equation}\label{eq3}
\beta^{(i-j)m_1}=1\Longleftrightarrow(i-j)(\frac{q^{\frac{m-1}{2}}+1}{2})\equiv0\pmod {\frac{q^m-1}{2}}.
			\end{equation}
Since $2\nmid m$ and $\gcd(\frac{m-1}{2},m)=1$, then $\gcd(q^{\frac{m-1}{2}}+1,q^{m}-1)=2$.
Thus Eq.\ref{eq3} holds $\Leftrightarrow$ $i=j$, a contradiction. 
Thus $ d({\C^{\perp}_{(n,1,\delta,1)}})\geq 3$. 
			
By Lemma \ref{lem4}, if $d(\C^{\perp}_{(n,1,\delta,1)})= 5$, then $\sum_{i=0}^{2}(q-1)^{i} \binom{n}{i}> q^{m+1}$ ,
this is infeasible, i.e., $d({\C^{\perp}_{(n,1,\delta,1)}})\leq 4$.
For $q=3$, if $d(\C^{\perp}_{(n,1,\delta,1)})= 4$, we have $\sum_{i=0}^{1}(3-1)^{i} \binom{n-1}{i}\leq 3^{m}$ satisfying Lemma \ref{lem4}.
 Thus $3\leq d(\C^{\perp}_{(n,1,\delta,1)})\leq 4$.
For $q>3$, if $d(\C^{\perp}_{(n,1,\delta,1)})= 4$, then $\sum_{i=0}^{1}(q-1)^{i} \binom{n-1}{i}> q^{m}$, this is infeasible,
i.e., $d(\C^{\perp}_{(n,1,\delta,1)})=3$. 			
\end{proof}

\begin{example}
	Let $(q,m)=(3,3)$. Then $\overline{\C}_{(n,1,\delta_1,1)}$ is an $[14,4,8]$ code and weight enumerator is
	$1+26z^8+26z^9+26z^{11}+2z^{14}$, which has the best known parameters for linear code in Database \cite{RefJ9}. 
	Let $(q,m)=(5,3)$. Then $\overline{\C}_{(n,1,\delta_1,1)}$ is an $[63,4,48]$ code and weight enumerator is
	$1+248z^{48}+124z^{50}+248z^{53}+4z^{63}$, in Database \cite{RefJ9}, the best known linear code over $\F_5$, 
	having length $63$ and dimension $4$, holds a minimum distance of $49$.
	\end{example}

\begin{theorem}\label{th30}
	Let $2\nmid m$ and $\delta =\delta _2$. Then $\overline{\C}_{(n,1,\delta_2,1)}$ is an 
$[\frac{q^m-1}{2}+1,2m+1,\delta_2+1]$ extended code and $\C^{\perp}_{(n,1,\delta_2,1)}$ is an
$[\frac{q^{m}-1}{2},\frac{q^{m}-1}{2}-2m-1,d]$ dual code,
where 
$$\begin{cases}3 \leq d\leq6,&\text{if $q=3;$}\\
	3 \leq d\leq5,&\text{if $5\leq q\leq 9$};\\
	3 \leq d\leq4 ,&\text{if $q\geq 11$} .\\\end{cases}$$
 Moreover, Table \ref{tab:3} gives the weight distribution of $\overline{\C}_{(n,1,\delta_2,1)}$.
		\end{theorem}
\begin{proof}

	According to the Theorem 8 in \cite{RefJ35}, we have 
$$\begin{aligned}
		\C_{(n,1,\delta_2,1)}=\{(Tr_{q^{m}/q}(a\beta^{-\delta_{1}i}+b\beta^{-\delta_{2}i})+c)_{i=0}^{n-1}: a,b\in \F_{q^{m}},c\in \F_{q}\}\\
		 \end{aligned},$$ 
	its weight distribution is equal to
	\begin{equation}
	\left\{\bm v_4(a,b,c)=Tr_q^{q^m}(a\alpha ^{(q^{\frac{m-1}{2}}+1)i}+b\alpha ^{(q^{\frac{m-3}{2}}+1)i}+c)_{i=0}^{n-1}:a,b\in \F_{q^{m}},c\in \F_{q}\right\}.
	\end{equation}	
	Similar to Theorem \ref{th27}, 
	 we have $$\sum_{i=0}^{n-1}(Tr_q^{q^m}(a\alpha ^{(q^{\frac{m-1}{2}}+1)i}+b\alpha ^{(q^{\frac{m-3}{2}}+1)i})+c)=\frac{q-1}{2}c$$	
	for any codeword $\bm v_4(a,b,c)$ of $(4)$.
	And $\overline{\bm v_4}(a,b,c)=(\bm v_4(a,b,c), c)$ is the codeword of $\overline{\C}_{(n,1,\delta_2,1)}$,
 where $a,b \in \F_{q^m}$, $c \in \F_q$.
	Then $wt(\overline{\bm v_4}(a,b,c))=wt(\bm v_4(a,b,c))+\varepsilon$ where
	$$\varepsilon =\begin{cases}0,&\text{if $c=0;$}\\1,&\text{if $c\neq 0.$}\end{cases}$$ 
	Thus, we give the weight distribution in Table \ref{tab:3}.	
	\begin{table} 
		\caption{}
		\label{tab:3}
		\begin{center}
			\begin{tabular}{cc}
				\hline
				Weight & Frequency\\
				\hline
				0&1\\
				$\frac{q^m-q^{m-1}-q^{\frac{m+1}{2}}+1}{2}$ & $\frac{(q^m-1)(q^{m-1}-1)}{2(q+1)}$\\
				$\frac{(q-1)(q^{m-1}-q^{\frac{m-1}{2}})}{2}$ & $\frac{(q^m-1)(q^{m-1}+q^{\frac{m-1}{2}})}{2}$\\
				$\frac{q^m-q^{m-1}-q^{\frac{m-1}{2}}+1}{2}$ & $\frac{(q^m-1)(q^{m+2}-q^m-q^{m-1}-q^{\frac{m+3}{2}}+q^{\frac{m-1}{2}}+q^2)}{2(q+1)}$\\
				$\frac{q^m-q^{m-1}}{2}$ & $(q^m-1)(q^m-q^{m-1}+1)$\\
				$\frac{q^m-q^{m-1}+q^{\frac{m-1}{2}}+1}{2}$ &$\frac{(q^m-1)(q^{m+2}-q^m-q^{m-1}+q^{\frac{m+3}{2}}-q^{\frac{m-1}{2}}+q^2)}{2(q+1)}$\\
				$\frac{(q-1)(q^{m-1}+q^{\frac{m-1}{2}})}{2}$ & $\frac{(q^m-1)(q^{m-1}-q^{\frac{m-1}{2}})}{2}$\\
				$\frac{q^m-q^{m-1}+q^{\frac{m+1}{2}}+1}{2}$ & $\frac{(q^m-1)(q^{m-1}-1)}{2(q+1)}$\\
				$\frac{q^m+1}{2}$ & $q-1$\\
				\hline
			\end{tabular}
		\end{center}
	\end{table}
	
	For code $\C^{\perp}_{(n,1,\delta_2,1)}$, we have $T^{\perp}=(\Z_{n}\backslash T)^{-1}=C_{0}^n\cup C_{CL(n-\delta_1)}^n\cup C_{CL(n-\delta_2)}^n$. 
			Then $d(\C^{\perp}_{(n,1,\delta_2,1)})\geq2$. 
			
Note that $n-\delta_2=\frac{q^{m-1}+q^{\frac{m+1}{2}}}{2}$,
 then $CL(n-\delta_2)=\frac{q^{\frac{m-3}{2}}+1}{2}$ and 
$$C_{\frac{q^{\frac{m-3}{2}}+1}{2}}^n=\{\frac{q^{l}+q^{\frac{m-3}{2}+l}}{2}: 1\leq l \leq \frac{m+3}{2}\}\cup\{\frac{q^{l}+q^{l-\frac{m+3}{2}}}{2}: \frac{m+5}{2}\leq l \leq m\}.$$
Set $m_{2}=\frac{q^{\frac{m-3}{2}}+1}{2}$, thus $|C_{m_2}^n|=m$, i.e., $\dim(\C^{\perp}_{(n,1,\delta_2,1)})=n-2m-1$.

Suppose $d(\C^{\perp}_{(n,1,\delta_2,1)})=2$, then we have $a(x)=a_0+a_1x+a_2x^2+\dots+a_ix^i+\dots +a_jx^j+\dots+a_{n-1}x^{n-1}\in \C^{\perp}_{(n,1,\delta_2,1)}$ and 
$b(x)=b_0+b_1x+b_2x^2+\dots+b_ix^i+\dots +b_jx^j+\dots+b_{n-1}x^{n-1}\in \C^{\perp}_{(n,1,\delta_2,1)}$, where $a_t=b_t$ for any 
$t\in[0,n-1]\setminus \{i,j\}$ and $0\leq i<j\leq n-1$.
Since $a(1)=b(1)=0$, $a(\beta^{m_1})=b(\beta^{m_1})=0$ and $a(\beta^{m_2})=b(\beta^{m_2})=0$,
i.e.   $$\begin{cases} (a_i-b_i)+(a_j-b_j)=0, \\ (a_i-b_i)\beta^{im_1}+(a_j-b_j)\beta^{jm_1}=0,
	\\ (a_i-b_i)\beta^{im_2}+(a_j-b_j)\beta^{jm_2}=0,\end{cases}$$
	which equals to 
$$\beta^{(i-j)m_l}=1 \Leftrightarrow (i-j)m_l\equiv0\pmod {\frac{q^m-1}{2}}$$ for $l=1,2$.
			Obviously, we can get $ d({\C^{\perp}_{(n,1,\delta_2,1)}})\geq 3$.
		
For $3\leq q \leq 9$, if $d(\C^{\perp}_{(n,1,\delta_2,1)})=7$, then $\sum_{i=0}^{3}(q-1)^{i} \binom{n}{i}> q^{2m+1}$ by Lemma \ref{lem4},
this is infeasible, i.e., $3\leq d({\C^{\perp}_{(n,1,\delta_2,1)}})\leq 6$. For $q\geq 11$, if $d(\C^{\perp}_{(n,1,\delta_2,1)})=5$, we have 
$\sum_{i=0}^{2}(q-1)^{i} \binom{n}{i}> q^{2m+1}$ by Lemma \ref{lem4},
this is infeasible, i.e., $3\leq d({\C^{\perp}_{(n,1,\delta_2,1)}})\leq 4$. The cases of $5\leq q \leq 9$ and $q=3$ are similar and we omit it here.
\end{proof}

\begin{theorem}\label{th31}
	Let $2\mid m$ and $\delta =\delta _2$. Then $\overline{\C}_{(n,1,\delta_2,1)}$ is an  
    $[\frac{q^m-1}{2}+1,\frac{3m}{2}+1,\delta_2+1]$ extended code and $\C^{\perp}_{(n,1,\delta,1)}$ is an
    $[\frac{q^{m}-1}{2},\frac{q^{m}-1}{2}-\frac{3m}{2}-1,2\leq d\leq 4]$ dual code
     except for $(q,m)=(3,2)$. Moreover, Table \ref{tab:4} gives the weight distribution of $\overline{\C}_{(n,1,\delta_2,1)}$.
\end{theorem}
\begin{proof}
	If $2\mid m$, according to the Theorem 8 in \cite{RefJ35} for $\C_{(n,1,\delta_2,1)}$,
	its weight distribution is equal to\begin{equation}
	\left\{\bm v_5(a,b,c)=Tr_q^{q^{\frac{m}{2}}}(a\alpha ^{(q^{\frac{m}{2}}+1)i}+b\alpha ^{(q^{\frac{m-2}{2}}+1)i}+c)_{i=0}^{n-1}:a\in \F_{q^{\frac{m}{2}}},b\in \F_{q^{m}},c\in \F_{q}\right\}.
	\end{equation}	
	With the same way as Theorem \ref{th30},
	we directly give the weight distribution in Table \ref{tab:4}.

Note that $n-\delta_2=\frac{q^{m-1}+q^{\frac{m}{2}}}{2}$,
 then $CL(n-\delta_2)=\frac{q^{\frac{m}{2}-1}+1}{2}$ and 
$$C_{\frac{q^{\frac{m}{2}-1}+1}{2}}^n=\{\frac{q^{l}+q^{\frac{m}{2}+l-1}}{2}: 1\leq l \leq \frac{m}{2}+1\}\cup\{\frac{q^{l}+q^{l-\frac{m}{2}-1}}{2}: \frac{m}{2}+2\leq l \leq m\}.$$
Thus $|C_{\frac{q^{\frac{m}{2}-1}+1}{2}}^n|=m$. Then $\dim(\C^{\perp}_{(n,1,\delta_2,1)})=n-|T^{\perp}|=n-\frac{3m}{2}-1$.
			
For $ q=3$ and $m=2$, we can get  $\dim(\C^{\perp}_{(n,1,\delta_2,1)})=0$.
Otherwise, if $d(\C^{\perp}_{(n,1,\delta_2,1)})=5$, then
$\sum_{i=0}^{2}(q-1)^{i}\binom{n}{i}>q^{\frac{3m}{2}+1}$ by Lemma \ref{lem4}, this is impossible, 
i.e., $2\leq d({\C^{\perp}_{(n,1,\delta_2,1)}})\leq 4$. 
\end{proof}
\begin{table} 
	\caption{}
	\label{tab:4}
	\begin{center}
		\begin{tabular}{cc}
			\hline
			Weight&Frequency \\
			\hline
			0&1\\
			$\frac{q^m-q^{m-1}-q^{\frac{m}{2}}+1}{2}$ & $\frac{(q^m-1)(q^{\frac{m+2}{2}}+q^{\frac{m-2}{2}}-2)}{2(q+1)}$\\
			$\frac{(q-1)(q^{m-1}-q^{\frac{m}{2}-1})}{2}$ & $\frac{(q^m-1)(q^{\frac{m+2}{2}}+q)}{2(q+1)}$\\
			$\frac{q^m-q^{m-1}-q^{\frac{m-2}{2}}+1}{2}$ & $\frac{(q^{\frac{m}{2}}-1)(q^{m+1}-2q^m+q)}{2}$\\
			$\frac{q^m-q^{m-1}}{2}$ & $(q^m-1)q^{\frac{m-2}{2}}$\\
			$\frac{q^m-q^{m-1}+q^{\frac{m-2}{2}}+1}{2}$ &$\frac{(q^m-1)(q^{\frac{m+2}{2}}+q)(q-1)}{2(q+1)}$\\
			$\frac{(q-1)(q^{m-1}+q^{\frac{m}{2}-1})}{2}$ & $\frac{(q^{\frac{m+2}{2}}-q)(q^m-2q^{m-1}+1)}{2(q-1)}$\\
			$\frac{q^m-q^{m-1}+q^{\frac{m}{2}}+1}{2}$ & $\frac{(q^m-1)(q^{\frac{m}{2}}-q^{\frac{m-2}{2}})}{2}$\\
			$\frac{(q-1)(q^{m-1}+q^{\frac{m}{2}})}{2}$ & $\frac{(q^m-1)(q^{\frac{m-2}{2}}-1)}{(q^2-1)}$\\
			$\frac{q^m+1}{2}$ & $q-1$\\
			\hline
		\end{tabular}
	\end{center}
\end{table}
	
\begin{example}
   Let $q=3$, $m=4$. Then $\overline{\C}_{(n,1,\delta_2,1)}$ is an $[41,7,23]$ linear code and weight enumerator is
$1+280z^{23}+300z^{24}+336z^{26}+240z^{27}+600z^{29}+168z^{30}+240z^{32}+20z^{36}+2z^{40}$, which has the best known parameters for linear code in Database \cite{RefJ9}. 
Let $q=3$, $m=3$. Then $\overline{\C}_{(n,1,\delta_2,1)}$ is an $[14,7,5]$ code and weight enumerator is
$1+26z^{5}+156z^{6}+624z^{8}+494z^{9}+780z^{11}+78z^{12}+26z^{13}+2z^{14}$, in Database \cite{RefJ9},
the best known linear code over $\F_3$, having length $14$ and dimension $7$, holds a minimum distance of $6$.
\end{example}
	
	\subsection{Dually-BCH code}	
	If $\C$ and $\C^{\perp}$ are BCH codes with respect to $\beta $, 
	then $\C$ is a dually-BCH code \cite{RefJ8}.	
	For $b=2$, we present the necessary and sufficient condition for $\C_{(n,1,\delta,b)}$ to be recognized as a dually-BCH code.
	\begin{theorem}\label{th33}
		Let $q\geq 3$ and $m\geq3$ be odd. Then $\C_{(n,1,\delta,2)}$ is a dually-BCH code
		iff $\delta _1-1 <\delta \leq n-1$.
	\end{theorem}
	
	\begin{proof}
		For $\C_{(n,1,\delta,2)}$ with $\delta \in[2,n-1]$, we have $T=C_2^n\cup C_3^n\cup \dots \cup C_\delta^n$ . 
		Since $0 \notin T$, then $C_0^n \subset T^{\perp}$.
		Then $\C_{(n,1,\delta,2)}$ is a dually-BCH code $\Leftrightarrow $ there exists $t\geq 1$ such that 
		$T^{\perp}=C_0^n\cup C_1^n\cup \dots \cup C_{t-1}^n$.
		
		For $q \leq \delta \leq n-1$, we have $C_1^n=C_q^n$, then $T=C_1^n\cup C_2^n\cup \dots \cup C_\delta^n$. 
		According to Theorem 30 of \cite{RefJ29},
		we deduce that $\C_{(n,1,\delta,2)}$ will be a dually-BCH code iff $\delta_1-1< \delta \leq n-1$.
		
		For $2 \leq \delta \leq q-1$, we know that $|C_i^n|=m$ for any $1\leq i \leq q-1$ and $q\nmid i$ by Lemma \ref{lem5}. 
		Since $C_1^n\notin T$ and $C_2^n\in T$, then $\dim(\C_{(n,1,\delta,2)})\in [m,n-m]$. 
		Note that $\frac{q^{\frac{m-1}{2}}+1}{2}\in \Gamma _n$. We consider two cases.
		
		(1) If $m\geq 5$, we have $\frac{q^{\frac{m-1}{2}}+1}{2}>q-1>\delta $,
	then $\frac{q^{\frac{m-1}{2}}+1}{2}\notin T$.
		Since $$CL(n-\frac{q^{\frac{m-1}{2}}+1}{2})=\delta _1,$$ then $C_{\delta _1}^n\subset T^{\perp}$.
		If $T^{\perp}=C_0^n\cup C_1^n\cup \dots \cup C_{\delta _1}^n$, a contradiction.
		
		(2) If $m=3$, we discuss this in two parts. 
		
		For $q>3$. If $2 \leq \delta \leq \frac{q-1}{2}$, 
		then $\frac{q^{\frac{m-1}{2}}+1}{2}>\frac{q-1}{2}>\delta$, i.e., $\frac{q^{\frac{m-1}{2}}+1}{2}\notin T$.
		By the same way as (1), $\C^{\perp}_{(n,1,\delta,2)}$ is not a BCH code. 
		If $\frac{q+1}{2} \leq \delta \leq q-1$, then $\frac{q+1}{2}\in T$. Note that $C_1^n\notin T$ 
		and $T^{\perp}={(\Z_n \setminus T)}^{-1}$, then $CL(n-1)=CL(\frac{q^3-3}{2})=\frac{q^3-2q^2-1}{2}=\delta_2\in T^{\perp}$. 
		Thus, $\C_{(n,1,\delta,2)}$ is a dually-BCH code $\Rightarrow $ $C_0^n\cup C_1^n\cup \dots \cup C_{\delta _2}^n\in T^{\perp}$, 
		i.e., $\dim(\C^{\perp}_{(n,1,\delta,2)})\leq m$. Since $q>3$, and $2$, $\frac{q+1}{2}\in T$ are coset leader by Lemma \ref{lem5}, 
		then $\dim(\C_{(n,1,\delta,2)})\leq n-2m$, a contradiction.

		For $q=3$, then $\delta=2$. We can easily conclude that $\C_{(13,1,2,2)}$ does not satisfy the conditions 
		for being a dually-BCH code.
		\end{proof}
	
	\begin{theorem}\label{th34}
		Let $q\geq 3$ and $m\geq2$ be even. Then $\C_{(n,1,\delta,2)}$ is a dually-BCH code
		iff $\delta _1-1 < \delta \leq n-1$.
	\end{theorem}
	\begin{proof}
		For $\C_{(n,1,\delta,2)}$ with $\delta \in[2,n-1]$, we have $T=C_2^n\cup C_3^n\cup \dots \cup C_\delta^n$ . 
		Since $0 \notin T$, then $C_0^n \subset T^{\perp}$.
		Then $\C_{(n,1,\delta,2)}$ is a dually-BCH code $\Leftrightarrow $ there exists $t\geq 1$ such that 
		$T^{\perp}=C_0^n\cup C_1^n\cup \dots \cup C_{t-1}^n$.
		
		For $q \leq \delta \leq n-1$, we have $C_1^n=C_q^n$, then $T=C_1^n\cup C_2^n\cup \dots \cup C_\delta^n$. 
		According to Theorem 30 of \cite{RefJ29},
		we deduce that $\C_{(n,1,\delta,2)}$ is a dually-BCH code iff $\delta _1-1 < \delta \leq n-1$.
		
		For $2 \leq \delta \leq q-1$, we know that $|C_i^n|=m$ for any $1\leq i \leq q-1$ and $q\nmid i$ by Lemma \ref{lem6}.
		Since $C_1^n\notin T$ and $C_2^n\in T$, then $\dim(\C_{(n,1,\delta,2)})\in [m,n-m]$. Note that $\frac{q^{\frac{m}{2}}+1}{2}\in \Gamma _n$.
		We consider two cases:
		
		(1) If $m\geq 4$, we have $\frac{q^{\frac{m}{2}}+1}{2}>q-1>\delta $,
		then $\frac{q^{\frac{m}{2}}+1}{2}\notin T$. Since
		$$CL(n-\frac{q^{\frac{m}{2}}+1}{2})=\delta _1,$$ then $C_{\delta _1}^n\subset T^{\perp}$.
		If $T^{\perp}=C_0^n\cup C_1^n\cup \dots \cup C_{\delta _1}^n$, a contradiction.

		(2) 
		If $m=2$, we discuss this in two parts. 

		For $q>3$. If $2 \leq \delta \leq \frac{q-1}{2}$, 
		then $\frac{q^{\frac{m}{2}}+1}{2}>\frac{q-1}{2}>\delta$, i.e., $\frac{q^{\frac{m}{2}}+1}{2}\notin T$.
		By the same way as (1), $\C^{\perp}_{(n,1,\delta,2)}$ is not a BCH code. 
		If $\frac{q+1}{2} \leq \delta \leq q-1$, then $\frac{q+1}{2}\in T$. Note that $C_1^n\notin T$ and $T^{\perp}={(\Z_n \setminus T)}^{-1}$, 
		then $CL(n-1)=CL(\frac{q^2-3}{2})=\frac{q^2-2q-1}{2}=\delta_2\in T^{\perp}$. 
		Thus, $\C_{(n,1,\delta,2)}$ is a dually-BCH code $\Rightarrow $ $C_0^n\cup C_1^n\cup \dots \cup C_{\delta _2}^n\in T^{\perp}$, 
		i.e., $\dim(\C^{\perp}_{(n,1,\delta,2)})\leq m$. Since $q>3$, 
		and $2$, $\frac{q+1}{2}\in T$ are coset leader by Lemma \ref{lem6}, then $\dim(\C_{(n,1,\delta,2)})\leq n-2m$,
		a contradiction.

		For $q=3$, then $\delta=2$. Thus we get a code $\C_{(4,1,2,2)}$, which is not a dually-BCH code obviously. 	
	\end{proof}

\section{Conclusions}\label{set} The principal contributions outlined in this paper are:
		
\begin{itemize}
\item For the negacyclic BCH code of length $n=\frac{q^{m}-1}{2}$, we get  
$\dim(\C_{(n,-1,\left\lceil \frac{\delta+1}{2}\right\rceil,0)})$ for $\delta =a\frac{q^m-1}{q-1}$, $aq^{m-1}-1$($1\leq a <\frac{q-1}{2}$) (see Lemmas \ref{lem8}, \ref{lem10})
and $\delta=a\frac{q^m-1}{q-1}+b\frac{q^m-1}{q^2-1}$, $aq^{m-1}+(a+b)q^{m-2}-1$ 
$(2\mid m,1\leq a+b \leq q-1$,$\left\lceil \frac{q-a-2}{2}\right\rceil\geq 1)$ (see Lemmas \ref{lem11}, \ref{lem13}). 
And we get $\dim(\C_{(n,-1,\delta,0)})$ with few nonzeros (see Theorems \ref{th17}, \ref{th19}, \ref{th21}) 
and  $\dim(\C_{(n,-1,\delta,b)})$ with $b\neq 0$ are settled (see Theorem \ref{th23}).
\item For the cyclic BCH code of length $n=\frac{q^{m}-1}{2}$, we give the weight distribution of 
extended code $\overline{\C}_{(n,1,\delta,1)}$ and the parameters of dual code $\C^{\perp}_{(n,1,\delta,1)}$ 
for $\delta_2\leq \delta \leq \delta_1$ (see Theorems \ref{th27}, \ref{th28}, \ref{th30}, \ref{th31}).
Moreover, we present the necessary and sufficient condition for $\C_{(n,1,\delta,2)}$ is a dually-BCH code
(see Theorems \ref{th33}, \ref{th34}).
 \end{itemize}	
\section*{Declaration of competing interest}
The authors declare that they have no known competing financial interests or personal relationships that 
could have appeared to influence the work reported in this paper.
		
\section*{Data availability}
Data will be made available on request.
\section*{Acknowledgements}
 This research was supported by the National Natural Science Foundation of China (No.U21A20428 and 12171134).

		
	\end{document}